\documentclass[a4paper]{article}
\usepackage[english]{babel}
\usepackage[utf8]{inputenc}
\usepackage{amsmath}
\usepackage{amsthm}
\usepackage{amssymb}
\usepackage{breqn}
\usepackage{graphicx}
\usepackage{fancyhdr}
\usepackage{subfigure}
\usepackage{subfig}
\usepackage{booktabs}
\usepackage{setspace}
\usepackage{caption}
\usepackage[affil-it]{authblk}
\usepackage[colorinlistoftodos]{todonotes}
\usepackage[margin=1in]{geometry}
\usepackage[T1]{fontenc}

\newtheorem{theorem}{Theorem}

\theoremstyle{plain}

\theoremstyle{definition}

\theoremstyle{remark}

\title{Value-at-Risk:\\
 The Effect of Autoregression in a Quantile Process\\}

\author{Khizar Qureshi \thanks{Generous amounts of feedback from: \newline Professor Alfred Galichon, Professor Victor Chernozhukov \newline Department of Economics, \newline Massachusetts Institute of Technology}}
\affil{Massachusetts Institute of Technology \footnote{Department of Mathematics, Department of Chemical Engineering, Year: 2016 \\ }}

\date{\today}
\fancypagestyle{phdthesis}{%
\fancyhf{}
\fancyhead[R]{\thepage}
}

\begin{document}
\onehalfspacing 
\maketitle

\begin{abstract}
Value-at-Risk (VaR) is an institutional measure of risk favored by financial regulators. VaR may be interpreted as a quantile of future portfolio values conditional on the information available, where the most common quantile used is $95 \%$. Here we demonstrate Conditional Autoregressive Value at Risk, first introduced by Engle, Manganelli (2001). CAViaR suggests that negative/positive returns are not i.i.d., and that there is significant autocorrelation. The model is tested using data from 1986-1999 and 1999-2009 for GM, IBM, XOM, SPX, and then validated via the dynamic quantile test. Results suggest that the tails (upper/lower quantile) of a distribution of returns behave differently than the core. 
\end{abstract}

\section*{Introduction}
Several recent financial disasters have made clear the necessity for a diverse set of risk management tools. Traditional models of risk management often rely on trivial probabilistic tools and often fail to relax key assumptions for the underlying statistics. An effective tool for risk management should be a withstanding measure of uncertainty robust to a large set of situations. Moreover, it should be suitable for its users, adaptive to complex situations, and compatible for various sample sizes. Risk management is not simply a tool to establish the upper bound on a loss, but also a preventative measure that should lead to the development of an informed decision-making process.
\paragraph{}
Perhaps the most well known tool for risk management amongst finance practitioners is Value at Risk (VaR). Conceptually, VaR measures the supremum of a portfolio's loss with a particular level of confidence. Consider a portfolio of unitary value with an annual standard deviation of  $15 \%$. The $95 \%$ daily VaR is simply the product of the daily standard deviation and the total value, or $\$1 \times 2.35 \times 0.15 = \$0.3525$. A $95 \%$ daily VaR of 0.3525 means that, if our day were hypothetically conducted an infinite number of times, the loss of our one dollar portfolio would be greater than 0.3525 with probability 0.05.

We say with $95 \%$ confidence that on a given day, the maximum loss implied by VaR is 0.3525 for the one dollar portfolio. 
\paragraph{}
A more rigorous definition of VaR is a particular quantile of future portfolio values, conditional on current information. In particular, we say that $P \left( y_t <  VaR_t | \Omega_{t-1} \right) = \alpha$, where $y_t$ is a time t return, $\Omega$ is the set of available information in a weak sense, and $\alpha$ is the confidence level or probability. The immediate considerations for a functional model include a closed-form representation, a set of well-defined intermediary parameters, and a test to validate the proposed model. In advance of our model proposition(s), we will review and evaluate existing models for VaR.
\paragraph{}
The remainder of the paper is organized as follows. Section II will introduce and evaluate existing models for Value at Risk, all of which will guide us in constructing CAViaR. Section III will cover the notion of Conditional Autoregression, the understanding of which is critical to realistic non-i.i.d processes. Section IV will introduce various methods of testing quantile regression, which will enable us to compare the set of well-known models with ours. Section V will focus on the empirical test of CAViaR on IBM, GM, and SPX time series data. We will conclude with section VI. 
\section*{Existent Models}
VaR has become a quintessential tool for portfolio management because it enables funds to estimate  the cost of risk and efficiently allocate it. Moreover, a growing number of regulatory committees now requires  institutions to monitor and  report VaR  frequently. Such a measure discourages excessive leverage and increases transparency of the "worst-case  scenario". While VaR , methods of estimation vary across both markets and firms. For ease and convenience of terminology, we will refer to all institutions and funds concerned with monitoring VaR as "holders".
\paragraph*{}
One component that often varies between different VaR models is the method by which the distribution of portfolio returns is estimated. A rudimentary example was readily introduced at the beginning of the paper, in which returns were assumed to be independently and identically distributed, or i.i.d. As we will see, however, returns almost never follow a martingale process, but rather, are Markov. A portfolio's performance on any given day almost always effects the performance on subsequent days. Thus, the probability of observing a specific return or variance as an event is dependent on the probability of observing the same event one period prior. The calculation of returns falls within two categories:
\begin{enumerate}
\item {\bf Factor Models:} Here, a universe of assets are studied for their factors, all of which are correlated. Thus, the time variation in the risk of the portfolio is derived from the volatility of the correlations. A well known example is the Fama-French four-factor model. The approach, however, assumes that negative returns follow the same approach as non-negative returns. Perhaps an even more alarming assumption by such models is the homoscedasticity between returns per unit risk.
\item {\bf Portfolio Models:} Here, VaR is instantaneously constructed using statistical inference of past portfolios. Then, quantiles are forecast via several approaches, including Generalized Autoregressive Conditional Heteroscedasticity (GARCH), exponential smoothing, etc., most of which incorrectly assume normality. Moreover, the set of models assume that after a certain amount of time, a particular historical return has probability zero of recurring. 
\end{enumerate}
We can deduce without an empirical demonstration that that portfolio models will underestimate VaR after time T, and factor models will fail to account for autoregression. Interestingly enough, the last decade has motivated the introduction of extreme quantile estimation, and the notion of asymptotic tail distributions. Many of these models, however, are only representative of especially low ($< 1 \%$) quantiles, and do not relax the weak i.i.d. assumption.
\section*{Conditional Autoregressive Value-at-Risk}
We now address many of the concerns above with CAViaR. In particular, we will study the asymptotic distribution, account for autocorrelation, and do so under various regimes. Suppose that there exists an observable vector of returns, $\{y_t \}_{t=1}^T$. We denote $\theta$ as the probability associated with the Value-at-Risk. Letting $x_t$ be a vector of time t observable variables (i.e. returns), and $\beta$ be a vector of unknown parameters, a Conditional Autoregressive VaR model may take the following form:

\begin{equation}
f_t (\beta) = \beta_0 + \sum_{i=1}^q \beta_i f_{t-i} (\beta) + \sum_{j=1}^r \beta_j l(x_{t-j})
\end{equation}
{\bf Interpretation:}\\
The quantile of portfolio returns at a time t is a function of not only the past period returns, but also the past period quantile of returns. That is, $f_t(\beta) = f_t(x_{t-1}, \beta_{\theta}$. The lag operator in the third term, l, links the set of available information at ${t-j}$ to the quantile of returns at t. The autoregressive function, $f_{t-1} (\beta)$, creates a smooth path between time-oriented quantiles. The first term, $\beta_0$, is simply a constant. The example provided at the beginning of the paper, which does not account for autoregression, would simply remain a constant: $f_t(\beta) = \beta_0$. Now that we have developed an understanding of the basic form for a CAViaR model, we will explore a few examples.
\subsection*{Adaptive Model}
In general, an adaptive model follows the form
\begin{equation}
f_t(\beta_1) = f_{t-1}(\beta_1) + \beta_1 \{[1+exp(G(y_{t-1}-f_{t-1}(\beta_1))))]^{-1}-\theta \}
\end{equation}
The adaptive model successfully accounts for increase in expected VaR. Whereas the traditional model would change only with a change in portfolio value, the adaptive model increases the Value-at-Risk by unit one whenever it is exceeded. Moreover, it decreases VaR by unit one if initial estimates proved to be too high. It is clear to see that such a conditional adjustment, in the form of a step function, would provide for a more accurate myopic estimation. However, because all changes are of magnitude one, the adaptive model overlooks large deviations in returns upwards or downwards. For example, consider a state in which the portfolio halved in value for three consecutive days. While the portfolio has been left at an eighth of its value, the VaR only increased by three units from its value at t=0. 
\subsection*{Symmetric Absolute Value}
The SAV model takes the form
\begin{equation}
f_t(\beta) = \beta_1 + \beta_2 f_{t-1} (\beta) + \beta_3 |y_{t-1}|
\end{equation}
The SAV is an autoregressive model in which a change in returns, regardless of direction, results in a change in VaR. What is particularly useful about this model is its ability to generalize movement in portfolio value. However, it is because of this very feature that SAV should not be used as a primary tool for measurement. Consider a series of large deviations in portfolio value, alternating upwards and downwards. While the long-run change in value is zero, the VaR implied by SAV would be unrealistically high. Similarly, a series of small deviations would imply an unrealistically low VaR. 
\subsection*{Asymmetric Slope}
The AS model takes the form
\begin{equation}
f_t(\beta) = \beta_1 + \beta_2 f_{t-1} (\beta) + \beta_3 (y_{t-1})^{+} + \beta_4(y_{t-1})^{-}
\end{equation}
The AS model is intended to capture the asymmetric leverage effect. Specifically, it was designed to detect the tendency for volatility to be greater following a negative return than a positive return of equal magnitude. The model relies on magnitude of error, rather than squared error, as in GARCH.
\subsection*{Indirect GARCH(1,1)}
The Indirect GARCH model takes the form
\begin{equation}
f_t(\beta) = (\beta_1 + \beta_2 f^2_{t-1} (\beta) + \beta_3 y^2_{t-1})^{\frac{1}{2}}
\end{equation}
While the GARCH model is estimated by maximum likelihood, Indirect GARCH is estimated via quantile regression. 
\section*{Regression Quantiles}
Thus far, we have understood the general form of a Conditional Autoregressive Value at Risk model, and have also seen several possible forms. The primary difference between any pair of CAViaR models is the organization and treatment of $\beta_i$, the regressive parameter. In the case of SAV, we were interested in a $\beta$ that reflected magnitude of change in portfolio returns, where in AS, we were interested in only extreme ends of a series of returns. However, how are the underlying parameters actually measured? Koenker and Basset (1978) introduced the notion of a sample quantile to a linear regression. Consider a sample of observations $y_1, \ldots, y_T$ generated by the linear model
\begin{equation}
y_t = x'_t \beta_0 + \epsilon \theta_t, \ \ Q_{\theta} (\epsilon \theta_t | x_t) = 0
\end{equation}
where $x_t$ is a length p vector of regressors (i.e. returns), and $Q_{\theta} (\epsilon \theta_t | x_t)$ is the $\theta$-quantile of $\epsilon \theta_t$ conditional on $x_t$. Consider the linear representation of an adaptive process: $f_t(\beta) = x_t \beta$. The $\theta$ regression quantile is $\hat{\beta}$ to satisfy the objective
\begin{equation}
min_{\beta} \frac{1}{T} \sum_{t=1}^T [\theta - I(y_t < f_t(\beta))][y_t - f_t(\beta)]
\end{equation}
Qualitatively, we adjust beta until VaR is no longer exceeded, or "hit" by a certain amount. Such a condition is satisfied when the observed indicator variables are 0. To account for the set of available information,
\begin{equation}
y_t = f(y_{t-1}, x_{t-1}, \ldots, y_1, x_1; \beta_0) + \epsilon_{t \theta}, \ \ [Q_{\theta} (\epsilon_{t \theta} | \Omega_t) = 0] = f_t(\beta^0) + \epsilon_{t \theta}, \ \ t=1,\ldots, T,
\end{equation}
\begin{theorem}[Consistency]
For generalized model (8), $\hat{\beta} \rightarrow \beta^0$ in probability, where $\hat{\beta}$ solves
\begin{equation}
min_{\beta} T^{-1} \sum_{t=1}^T \{[ \theta - I(y_t < f_t(\beta))] \times [y_t-f_t(\beta)] \}
\end{equation}
\end{theorem}
\begin{proof}
Please see Appendix
\end{proof}
\begin{theorem}[Asymptotic Normality]
In testing for asymptotic normality, for statistic T,
\begin{equation}
\sqrt{T} A_T^{-\frac{1}{2}} D_T (\hat{\beta} - \beta^0) \rightarrow N(0,1)
\end{equation} 
where
\begin{equation}
\begin{split}
A_T = E \left(T^{-1} \theta (1-\theta) \sum_{t=1}^T \nabla ' f_t (\beta^0) \nabla f_t (\beta^0) \right)\\
D_T = E \left(T^{-1} \sum_{t=1}^T h_t (0 | \Omega_t) \nabla ' f_t (\beta^0) \nabla f_t (\beta^0) \right)
\end{split}
\end{equation}
\end{theorem}
\begin{proof}
\thanks{Proof left as exercise for reader}
\end{proof}
\section*{Testing Quantile Models}
While the expanded set of models available now account for the autocorrelation of returns, as well as large deviations, they must be tested. Given a new observation, the model remains valid if $P[y_t < f_t (\beta^0)] = \theta$. If such a condition holds for the entirety of the time series, then it is proven valid. As shown by Christoffersen (1998), such a method is equivalent to testing for the independence of indicator variables for the same condition. In other words, $[I(y_t < f_t \left (\beta^0 \right))]_{t=1}^T = \{I_t \}_{t=1}^T$. While this provides for a natural test of forecasting models, it does not fully assess the validity of quantiles. To test conditional quantile models, we introduce a representative indicator variable that changes with the quantile itself. Define a sequence of independent random variables $\{z_t \}_{t=1}^T$ such that
\begin{itemize}
\item $\mathbb P(z_t = 1) = \theta$
\item $\mathbb P(z_t = -1) = 1-\theta$
\end{itemize}
Expressing positive or negative autocorrelated returns in terms of $z_t$ indeed accounts for the probability of exceeding a quantile. However, whilst the unconditional probabilities are uncorrelated, the conditional probabilities for a hit still depend on one another. Because these tests evaluate the lower bound of the VaR in the weakest sense, we work towards defining a dynamic quantile. Let
\begin{equation}
Hit_t(\beta^0) = I( y_t < f_t(\beta^0)) - \theta
\end{equation}
$Hit_t(\beta^0)$ assumes a value of $(1-\theta)$ for underestimations of VaR, and $-\theta$ otherwise. Notice that the expected value is zero, and that there should be no autocorrelation in the values between successive hits. 
\paragraph{}
For our first test, we determine whether the test statistic, $T^{-\frac{1}{2}}X'(\hat{\beta}) Hit(\hat{\beta})$ is significantly different from zero where $X_t(\hat{\beta}), t=1, \ldots, T$ may depend on $\hat{\beta}$, and is q-measurable information (i.e. returns). Suppose we wish to test the significance of an entire set of data along several $\beta$ simultaneously. Then, let $M_T = (X'(\beta^0)) - E[T^{-1}X'(\beta^0) H \nabla f(\beta^0)] D_T^{-1} \times \nabla' f(\beta^0)$. Here, H is a diagonal matrix with binary indicators conditioned on available information. Such a test would be run both in-sample, as well as out-of-sample. We will define and prove the conditions of each.
\begin{theorem}[In-Sample Dynamic Quantile Test]
If the assumptions made (see appendix) are valid, the following holds
\begin{equation}
[\theta (1-\theta) E(T^{-1} M_T M_T')]^{-\frac{1}{2}}T^{-\frac{1}{2}} X'(\hat{\beta}) Hit(\hat{\beta}) \sim N(0,I)
\end{equation}
Moreover,
\begin{equation}
DQ = \frac{Hit(\hat{\beta})X(\hat{\beta})(\hat{M}_T \hat{M'}_T)^{-1} X' (\hat{\beta}) Hit'(\hat{\beta}) }{\theta(1-\theta)} \sim \chi^2 \, \ T \rightarrow \infty
\end{equation}
where $\hat{M_T}$ is the difference between $X'(\hat{\beta})$ and a function of the gradient of $f(\hat{\beta})$.
\end{theorem}
\begin{proof}
\thanks{Proof left as exercise for reader}
\end{proof}

Essentially, the DQ test above tests whether or not the test statistic follows a normal distribution in the sense of an identity matrix, and whether the set of all dynamic quantiles in-sample follow a chi-squared distribution. We now shift our focus to a test statistic for dynamic quantiles out-of-sample.
\begin{theorem}[Out-of-Sample Dynamic Quantile Test]
Let $T_R$ denote the number of in-sample observations and $N_R$ denote the number of out-of-sample observations. Then
\begin{equation}
DQ = \frac{N_R^{-1} Hit' (\hat{\beta_{TR}}) [X(\hat{\beta_{TR}}) X(\hat{\beta_{TR}})]^{-1} \times X'(\hat{\beta_{TR}}) Hit' (\hat{\beta_{TR}})}{\theta (1-\theta)} \sim \chi^2
\end{equation}
\end{theorem}
\begin{proof}
\thanks{Proof left as exercise for reader}
\end{proof}
Use of the dynamic quantile tests allows for an estimation of the independence in "hits". The ideal quantile test would be one in which all hits $(y_t < VaR_t)$ are independent. Regulators would be able to choose between different measures of VaR when evaluating a portfolio.
\section*{Empirical Results}
The historical series of portfolio returns were studied for four different regimes of CAViaR. A total of 2,553 daily prices from WDRS/CRSP for General Motors (GM), International Business Machines (IBM), and the SP 500 (SPX) were retrieved. Daily returns were computed as 100 times the difference of the log of the prices. Two sets of ranges were used: one from April 7, 1986 to April 7, 1999, and the other from April 7, 1999 to June 1, 2009. From the total set of daily prices, 2,253 were used in-sample, and the last 300 were used out-of-sample. Both $1 \%$ and $5 \%$ day-end VaR were estimated for each of the four regimes.
\paragraph{}   
The optimization was completed using MATLAB R2015 and Gurobi Optimizer under a Quasi-Newton Method. The loops for recursive quantile functions (i.e. SAV) were coded in C.
\subsection*{Optimization Methodology}
Using a random number generator, n vectors were generated, each with uniform distribution in $[0,1]$. The regression quantile function was computed, and from the n vectors, $m \subseteq n$ vectors with the lowest regression quantile criteria were selected as initial values for optimization. For each of the four regimes, we first used the simplex algorithm. Following the approximation, we used a robust quasi-Newton method to determine new optimal parameters to feed into the simplex algorithm. This process was repeated until convergence, and tolerance for the regression quantile was set to $10^{-10}$.

\subsection*{Simplex Algorithm}
The algorithm operates on linear programs in standard form to determine the existence of a feasible solution. The first step of the algorithm (Phase 1) involves the identification of an extrema as a guess. Either a basic feasible solution is found, or the feasible region is said to be empty. In the second step of the algorithm (Phase 2), the basic feasible solution from Phase 1 is used as a guess, and either an optimal basic feasible solution is found, or the solution is a line with infinite (unbounded) optimal cost.

\subsubsection*{Quasi-Newton Method}
QN methods are used to locate roots, local maxima, or local minima if the Hessian is unavailable at each step. Rather, the Hessian is updated through analyzing gradient vectors. In general, a second order approximation is used to find a function minimum. Such a taylor series is
\begin{equation}
f(x_k + \Delta x) \approx f(x_k) + \nabla f(x_k)^T \Delta x + \frac{1}{2} \Delta x^T B \Delta x
\end{equation}
for a gradient $\Delta f$ and a Hessian approximate B, the gradient of the approximation is
\begin{equation}
\Delta f(x_k + \Delta x) \approx \nabla f(x_k) + B \Delta x
\end{equation}
with root
\begin{equation}
\Delta x = -B^{-1} \nabla f(x_k)
\end{equation}
We seek the Hessian $B_{k+1} = argmin_B \| B-B_k \|_V$ where V is a positive definite matrix defining the norm
\paragraph{}

\subsection*{CAViaR Results}
We now review Conditional Autoregressive Value-at-Risk, methodology introduced by (Engle, Manganelli, 2001). First we review results from 1986-1999. Then we extend to 1999-2012. To test the significance of autoregression, we study values of $\beta_{t-1}$, as well as p-values, for three regimes. Finally, we compare results of conditional VaR to unconditional VaR., the traditional risk measurement tool. 
\begin{figure}[h!]
\centering
\includegraphics[scale=1.0]{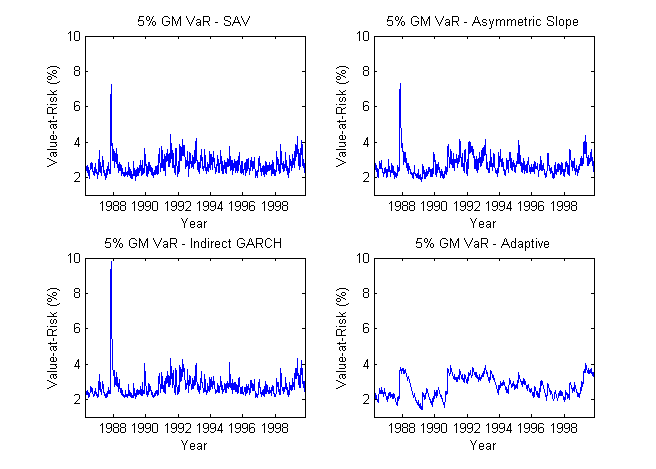}
\caption{The figure above shows $5 \%$ CAViaR plots for GM (1986-1999) under the following regimes: (a) Symmetric Absolute Value; (b) Asymmetric Slope; (c) Indirect GARCH; (d) Adaptive.} \end{figure}

\begin{figure}[h!]
\centering
\includegraphics[scale=0.8]{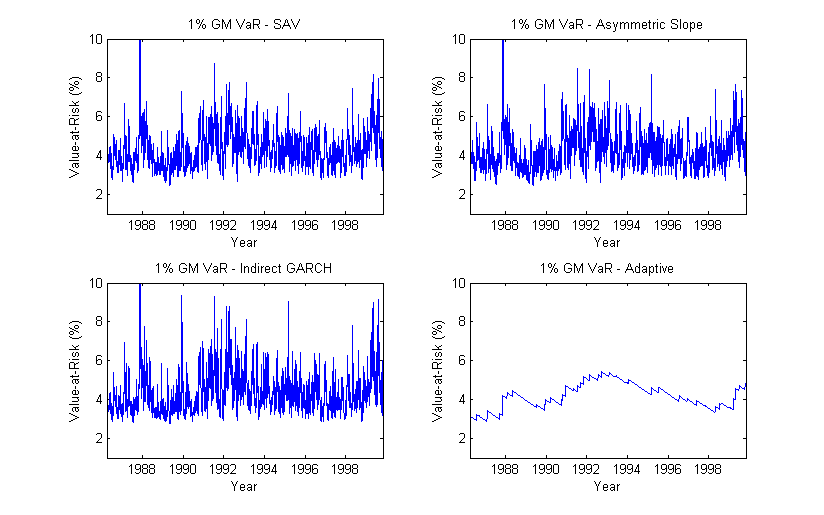}
\caption{The figure above shows $1 \%$ CAViaR plots for GM (1986-1999) under the following regimes: (a) Symmetric Absolute Value; (b) Asymmetric Slope; (c) Indirect GARCH; (d) Adaptive.} \end{figure}

\subsubsection*{Interpretation}
As previously discussed, the arrival of information creates uncertainty, and increases VaR. While this is true in both the conditional and unconditional case, it is emphasized in the former. It is well-known that the period between 1986 and 1999 was volatile, and was affected by events such as LTCM and Global crises in Russia and Asia. If VaR truly increases with the arrival of information, then it is sensible to see the peak in 1987, where all positive beta assets faced an increase in systemic risk with the market crash. The adaptive regime adjusts for changes in VaR, so the momentary increase in 1987 was given less weight.
\paragraph{}
It is also interesting to note that the autoregressive models capture non-systemic risk. While the market crash in 1987 affected all positive-beta assets, we also see less severe increases in VaR during 1999. With high probability, this is due to an idiosyncratic event--a total inventory recall by GM. Such a recall likely created uncertainty in expected cash flows, thus increasing the periodic volatility of the stock. The persistent volatility was exponentially weighted in the adaptive regime, resulting in a large increase towards the end of the period of study.

\begin{figure}[h!]
\centering
\includegraphics[scale=1.0]{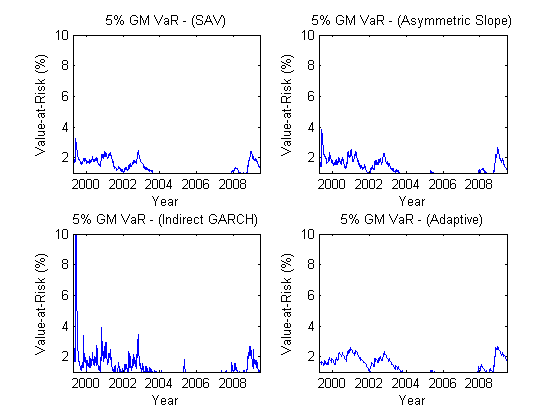}
\caption{The figure above shows $5 \%$ CAViaR plots for GM (1999-2012) under the following regimes: (a) Symmetric Absolute Value; (b) Asymmetric Slope; (c) Indirect GARCH; (d) Adaptive.}
\end{figure}

\begin{figure}[h!]
\centering
\includegraphics[scale=1.0]{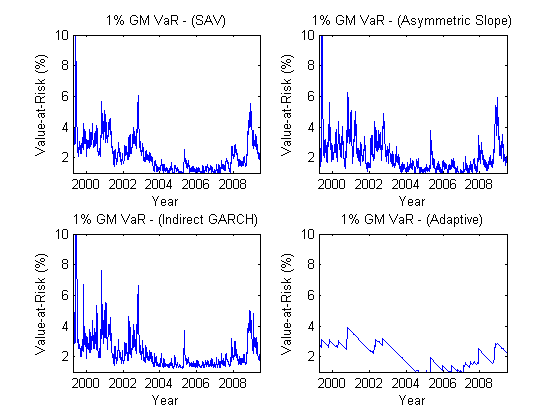}
\caption{The figure above shows $1 \%$ CAViaR plots for GM (1999-2012) under the following regimes: (a) Symmetric Absolute Value; (b) Asymmetric Slope; (c) Indirect GARCH; (d) Adaptive.}
\end{figure}

\subsubsection*{Interpretation}
Information in the 2000's was much more readily available than it was during 1980-2000. Consequently, a rapid digestion and reflection of information in asset prices may have resulted in more dynamic expectation. We see immediately that conditional VaR is much lower, indicating that the arrival of information did not induce as much uncertainty as it did in the decade prior. 
\paragraph{}
We are aware from the previous iteration that VaR increased in 1999 across all adaptive regimes. From this prior, it is sensible to see high VaR from the very first year of data, given that time is continuous. While the increased natural filtration of information suggests lower VaR, we infer sources based on backward-looking bias. Aside from the vehicle recall in 1999, the market faced a "mini-crash" in the early 2000's. This is better known as the "bubble burst". Further, the financial crisis in the 2008-2009 period resulted in an increase in conditional VaR. It is interesting to note the remarkable similarity across the four regimes, indicating an increase in the mean reversion coefficient of volatility.


\begin{figure}[h!]
\centering
\includegraphics[scale=0.7]{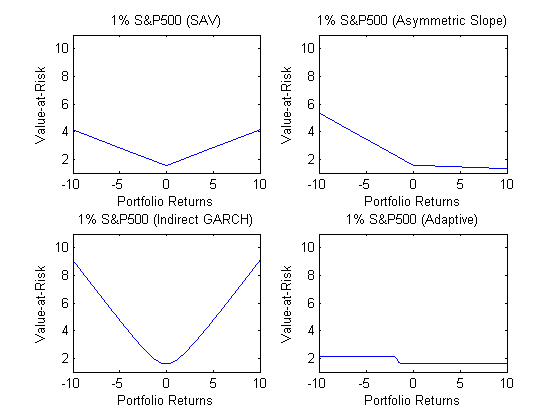}
\caption{The figure above shows the News Impact curve for $1 \%$ VaR estimates of the SP500. For $VaR_{t-1} = -2.576$, the curves show how $VaR_t$ changes with lagging portfolio returns (x). The AS curve suggests that negative returns may have a much stronger effect on the actual VaR estimate than do positive returns.}
\end{figure}

\begin{figure}[h!]
\centering
\includegraphics[scale=0.7]{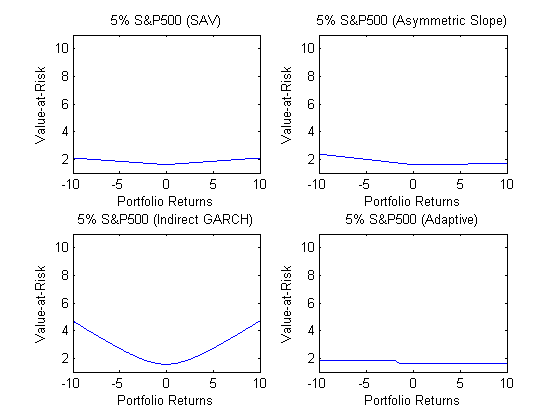}
\caption{The figure above shows the News Impact curve for $5 \%$ VaR estimates of the SP500. For $VaR_{t-1} = -1.960$, the curves show how $VaR_t$ changes with lagging portfolio returns (x). The AS curve suggests that negative returns may have a much stronger effect on the actual VaR estimate than do positive returns.}
\end{figure}

%

\begin{figure}[!htb]
\centering
\includegraphics[scale=0.6]{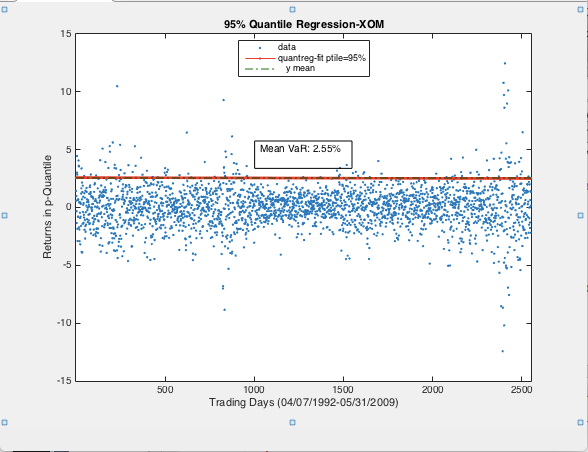}
\caption{The figure above shows the 95 $\%$ VaR for Exxon Mobil calculated through unconditional quantile regression.}
\end{figure}

\begin{figure}[!htb]
\centering
\includegraphics[scale=0.6]{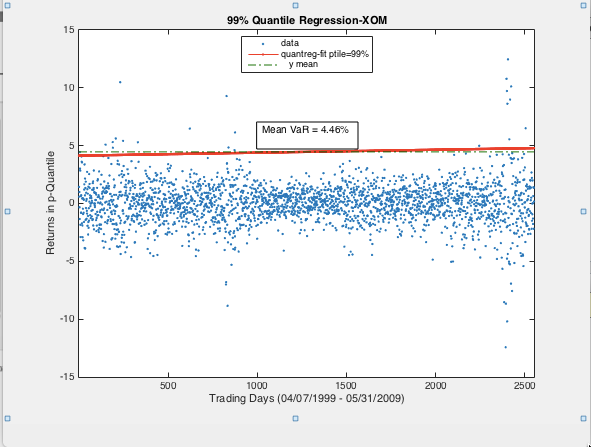}
\caption{The figure above shows the 99 $\%$ VaR for Exxon Mobil calculated through unconditional quantile regression. The VaR for quantile $p=0.99$ is higher than the VaR for quantile $p=0.95$ both in mean and in distribution. This is because the maximum loss must be determined with more certainty.}
\end{figure}

\begin{table}[!htb]
\centering
\begin{tabular}{@{}lllll@{}}
\toprule
               & AS     & SAV    & Indirect GARCH  \\ \midrule
$\beta_2, GM$  & 0.9160 & 0.9250 & 0.858               \\
p-value        & 0.0252 & 0.0186 & 0.0260                 \\
$\beta_2, IBM$ & 0.9178 & 0.9023 & 0.9068                 \\
p-value        & 0.0665 & 0.0582 & 0.0531               \\
$\beta_2, SPX$ & 0.9164 & 0.9252 & 0.8578                \\
p-value        & 0.0252 & 0.0186 & 0.026                 \\ \bottomrule
\end{tabular}
\caption{The table above shows the regressive parameter, $\beta_2$ for each asset under each of the three CAViaR regimes. With the exception of IBM by a slight margin, all three assets suggest significant autoregressive VaR at the $95 \%$ level for  (a) Asymmetric Slope and (b) Symmetric Absolute Value, and (c) Indirect GARCH}
\end{table}

\pagebreak
\subsection*{Broader Interpretation}
The arrival of news results in an expansion of information available. If the market is truly efficient, asset prices will reflect the expectations of those exposed to this information (Fama, 1997). However, expectations are not always dynamic, and integration of information may not be continuous. Consequently, volatility, a form of uncertainty, on a particular day will be autoregressive, and depend on volatility from previous days. It becomes necessary to use autoregression when calculating the expected loss within a $p-$quantile.
\paragraph{}
We demonstrate several adaptive models, including Symmetric Absolute Value, Asymmetric Slope, and Indirect GARCH (1,1). We recognize that while all of these models satisfy the requirements of autoregression, they differ in their treatment of $\beta_t$, the autoregressive parameter. When evaluating Value-at-Risk, conditional on information, we must carefully choose the model, and understand the underlying assumptions.
\paragraph{}
It is well known that for a monotonically increasing cumulative distribution function, an increase in $p$ will cover a larger portion of the distribution of risk. An immediate consequence of this is that the  $99 \%$ VaR exceeds the $95 \%$. We re-confirm the notion. We show that conditioning on the arrival of information may increase or decrease VaR, depending on the change in expectation. We contrast conditional against unconditional VaR, and autoregressive against independent VaR. In both contrasts, the former is a more proper treatment.   

\section*{Conclusion}
We have demonstrated the autoregressive properties of VaR conditioned on information available. Existing methods for calculating VaR fail to relax several assumptions including (i) the i.i.d. behavior of returns and (ii) the large deviations at the tail of the return distribution. CAViaR addresses both issues, and suggests that there is significant autocorrelation for the data used (GM, IBM, SPX 1999-2009). Regressive parameters $(\beta_i)$ were estimated by minimizing the regression quantile loss function, and the models were tested via dynamic quantiles. 
\paragraph{}
The worst performing method was the adaptive method, which failed to detect poor returns in 1999. The best performing method was Asymmetric Slope, which captured the larger effect of negative returns vs. positive returns on VaR. Symmetric Absolute Value (SAV), which does not segregate positive and negative returns, illustrated general portfolio movement. Indirect GARCH, which does not distinguish tails of distributions, overemphasized movement in 1999, 2001, and 2009. These very same results are reflected in the 1 $\%$ SPX News curves, which study the effect of news-driven returns on the SP 500.
\paragraph{}
Standard, non-conditional quantile regressions were studied for XOM, a large-cap stock similar to GM. Without autoregressive properties, 95 $\%$ VaR was underestimated relative to the CAViaR case. Such a result motivates the use of multiple risk measurement tools, each carrying different treatment of $(\beta_i)$ and underlying assumptions. 
\paragraph{}
 The most typical use of VaR involves determining the expected loss with 95 $\%$ certainty. While useful as a bare approximation, a more proper analysis of risk should be carried across various quantiles and multiple distributions. Moreover, the loss function for a given day should be clustered within time frames. We have re-confirmed that volatility is autoregressive and conditional on the arrival of information. Treating value at risk independently will almost surely underestimate maximum losses during periods of high volatility and overestimate during periods of low volatility. The use of multiple tools may be beneficial to financial institutions, the individuals they may represent, and the health of market participants as a whole. 
\subsection*{Extensions}
 Future applications include smaller time frames, different sets of stocks, an extension to the multivariate case, and quantile sensitivity. It is also worth investigating the implied volatility of out-of-the-money options calculated under the conditional and non-conditional regime. We expect the skew to carry more weight in the conditional regime to account for disaster risk. 
 \paragraph{}
 There also exists potential to improve the methodology used within the paper. Namely,
 \begin{itemize}
 \item Cost Regularization of Autoregressive parameters 
 \item Local inference of Hessian matrices to employ interior point methods for optimization
 \end{itemize}
 
 \pagebreak
\section*{Appendix}
\subsection*{Proofs}

\subsubsection*{Consistency}
\begin{proof}
Let $Q_T(\beta) = T^{-1} \sum_{t=1}^T q_t (\beta)$, where $q_t (\beta) = [\theta - I(y_t< f_t )\beta))][y_t - f_t (\beta)]$\\
{\bf Claim:} $E[q_t(\beta)]$ exists and is finite for every $\beta.$ This can be easily checked as follows:
\begin{equation}
\begin{split} 
&E[q_t(\beta)] < E|y_t - f_t(\beta) | \leq E| \epsilon_{t \theta} | + E |f_t(\beta)| + E|f_t(\beta^0)| < \infty\\
& \textit{Because f is continuous in $\beta$ (complete probability space), $q_t(\beta)$ must be continuous}\\
& E[V_T (\beta)] = E[Q_T(\beta) - Q_T(\beta^0)] \text{is uniquely minimized at $\beta^0$ for T sufficiently large}\\
& \text{Let } v_t (\beta) = q_t (\beta) - q_t (\beta^0), q_t(\beta) = [\theta - I(\epsilon_{t \epsilon} < \delta_t (\beta))][\epsilon_{t \epsilon} - \delta_t (\beta)] \text{for} \delta_t(\beta) = f_t(\beta) -f_t(\beta^0)\\
& E[v_t(\beta) | \Omega_t] = I(\delta_t(\beta) < 0) \int_{-|\delta_t(\beta)|}^{0} (\lambda + |\delta_t (\beta)|) h_t(\lambda | \Omega_t) d \lambda + I(\delta_t (\beta) >0) \int_0^{|\delta_t(\beta)|} (|\delta_t (\beta) - \lambda) h_t(\lambda | \Omega_t) d\lambda\\
&\text{We have $h_t (\lambda | \Omega_t) > h_t > 0$ whenever $| \lambda | < h_t$. Hence, for $0<\tau<h_1$}\\
& E[v_t (\beta) | \Omega_t] \geq I(\delta_t (\beta) < -\tau) \int_{- \tau}^0 [\lambda + \tau] h_1 d\lambda + I(\delta_t (\beta) > \tau) \int_{0}^{\tau} [-\lambda + \tau] h_1 d\lambda\\
& = \frac{1}{2} \tau^2 h_1 I( |\delta_t (\beta) | > \tau)\\
& \text{We take the unconditional expectation, and see } E[V_T (\beta)] = E[T^{-1} \sum_{t=1}^R v_t (\beta)] \geq \frac{1}{2} \tau^2 h_1 T^{-1} \sum_{t=1}^T P[ |f_t (\beta) - f_t (\beta^0) > \tau]
\end{split}
\end{equation}
\end{proof}

\subsection*{Consistency Assumptions}
\begin{enumerate}
\item $(\Omega, F, P)$ is a complete probability space, and $\{\epsilon_{t \theta}, x_t \}, t=1,2, \ldots$, are random vectors on this space.
\item The function $f_t (\beta): \mathcal R^{k_t} \times B \rightarrow \mathcal R$ is such that for each $\beta \in B$, a compact subset of $\mathcal R^p, f_t (\beta)$ is measurable with respect to the information set $\Omega_t$ and $f_t (.)$ is continuous in B, $t=1, 2, \ldots,$ for a given choice of explanatory variables $(y_{t-1}, x_{t-1}, \ldots, y_1, x_t \}.$
\item Conditional on all of the past information $\Omega_t$, the error terms $\epsilon_{t \theta}$ form a stationary process, with continuous conditional density $h_t (\epsilon | \Omega_t)$. 
\item There exists $h>0$ such that for all t, $h_t( 0| \Omega_t ) \geq h$.
\item $|f_t (\beta)| < K(\Omega_t)$ for each $\beta \in B$ and for all t, where $K(\Omega_t)$ is some (possibly) stochastic function of variables that belong to the information set, such that $E(|K(\Omega)|) \leq K_0 < \infty$ for some constant $K_0$.
\item $E[|\epsilon_{t \theta}|] <\infty$ for all t
\end{enumerate}

\subsection*{Asymptotic Normality Assumptions}
\begin{enumerate}
\item $f_t (\beta)$ is differentiable in B and for all $\beta$ and $\gamma$ in a neighborhood $v+0$ of $\beta^0$ such that $\| \beta - \gamma \| \leq d$ for $d$ sufficiently small and for all t:
\begin{enumerate}
\item $\| \nabla f_t (\beta) \| \leq F(\Omega_t)$, where $F(\Omega_t)$ is some (possibly) stochastic function of variables that belong to the information set and $E(F(\Omega_t)^3) \leq F_0 < \infty$, for some constant $F_0$.
\item $\| \nabla_t (\beta) - \nabla f_t(\gamma) \| \leq M(\Omega_t, \beta, \gamma) = O(\| \beta - \gamma \|),$ where $M(\Omega_t, \beta, \gamma)$ is some function such that $E[M(\Omega_t, \beta, \gamma)]^2 \leq M_0 \| \beta - \gamma \| < \infty$ and $E[M(\Omega_t, \beta, \gamma) F(\Omega_t)] \leq M_1 \| \beta - \gamma \| < \infty$ for some constants $M_0$ and $M_1$. 
\end{enumerate}
\item  $h_t(\epsilon | \Omega_t) \leq N < \infty \forall t$, for some constant N
\item $h(\epsilon | \Omega_t)$ satisfies the Lipschitz condition $|h_t (\lambda_1 | \Omega_t) -h_t(\lambda_2 | \Omega_t) | \leq L|\lambda_1 -\lambda_2 |$ for some constant $L <\infty \forall t$
\end{enumerate} 

\subsection*{Variance-Covariance Matrix Estimation Assumptions}
\begin{enumerate}
\item $\hat{c}_T/c_T \rightarrow 1$ in probability, where the non stochastic positive sequence $c_T$ satisfies $C_T = o(1)$ and $c^{-1}_T = o(T^{\frac{1}{2}})$
\item $E(|F(\Omega_t)|^4) \leq F_1 < \infty$ for all t and for some constant $F_1$, where $F(\Omega_t)$ has been defined under asymptotic normality
\item The difference between the DQ test and representative $A_T, D_T$ converge in probability to 0
\end{enumerate}

\pagebreak

\end{document}